\documentclass[a4paper,UKenglish,cleveref, autoref, thm-restate]{lipics-v2021}

\usepackage[noadjust]{cite} 
\usepackage{enumerate}
\usepackage{amsthm,complexity,algorithm,algpseudocode}
\usepackage{mathtools}
\usepackage{thmtools}

\usepackage{hyperref}

\usepackage{complexity}
\theoremstyle{plain}
\usepackage{amsmath,amssymb}
\usepackage{enumerate}
\usepackage{xspace}
\usepackage{listings}
\usepackage{color}
\usepackage{graphicx}
 \usepackage{xcolor}

\newcommand{\defparproblem}[4]{
  \vspace{3mm}
\noindent\fbox{
  \begin{minipage}{.95\textwidth}
  \begin{tabular*}{\textwidth}{@{\extracolsep{\fill}}lr} \textsc{#1}\\ \end{tabular*}
  {\bf{Input:}} #2  \\
  {\bf{Parameter:}} #3 \\
  {\bf{Question:}} #4
  \end{minipage}
  }
  \vspace{2mm}
}

\newcommand{\OO}{{\mathcal O}}

\newcommand{\HamDFVS}{{\sc Hamiltonian Cycle By DFVS}}
\newcommand{\HamPathDFVS}{{\sc Hamiltonian Path By DFVS}}
\newcommand{\LongPathG}{{\sc Longest Path Above Girth}}
\newcommand{\LongCycleG}{{\sc Longest Cycle Above Girth}}
\newcommand{\lef}{{\mathsf left}}
\newcommand{\ri}{{\mathsf right}}
\newcommand{\out}{{\mathsf out}}
\newcommand{\inn}{{\mathsf in}}

\newcommand{\sv}[1]{}

\newtheorem{reduction rule}{Reduction Rule}
\newtheorem{branching rule}{Branching Rule}
\usepackage{listings}
\usepackage{color}
\usepackage{complexity}
\usepackage{todonotes}
\newcommand{\todom}[1]{\todo[linecolor=blue,backgroundcolor=blue!25,bordercolor=blue]{#1}}

\newcommand{\mic}[1]{\textcolor{black}{#1}}

\title{Finding Long Directed Cycles Is Hard \\ Even When DFVS Is Small Or Girth Is Large}

\author{Ashwin Jacob}{Ben Gurion University of the Negev, Beersheba, Israel}{ashwinj@bgu.ac.il}{https://orcid.org/0000-0003-4864-043X}{}

\author{Micha{\l} W{\l}odarczyk}{Ben Gurion University of the Negev, Beersheba, Israel}{michal.wloda@gmail.com}{https://orcid.org/0000-0003-0968-8414}{}

\author{Meirav Zehavi}{Ben Gurion University of the Negev, Beersheba, Israel}{meiravze@bgu.ac.il}{https://orcid.org/0000-0002-3636-5322}{}
\date{}

\authorrunning{A. Jacob, M. W{\l}odarczyk, and M. Zehavi} 

\Copyright{Ashwin Jacob, Micha{\l} W{\l}odarczyk, and Meirav Zehavi} 



\funding{European Research Council (ERC) grant titled PARAPATH.}

\begin{document}
\maketitle

\begin{abstract}
We study the parameterized complexity of two classic problems on directed graphs: {\sc Hamiltonian Cycle} and its generalization {\sc Longest Cycle}. Since 2008, it is known that {\sc Hamiltonian Cycle} is W[1]-hard when parameterized by {\em directed treewidth} [Lampis et al., ISSAC'08]. By now, the question of whether it is FPT parameterized by the {\em directed feedback vertex set} (DFVS) number has become a longstanding open problem. In particular, the DFVS number is the largest natural \mic{directed} width measure studied in the literature.
In this paper, we provide a negative answer to the question, showing that even for the DFVS number, the problem \mic{remains} W[1]-hard.
As a consequence, we also obtain that {\sc Longest Cycle} is W[1]-hard on directed graphs when parameterized multiplicatively above girth, in contrast to the undirected case. This resolves an open question posed by Fomin et al.~[ACM ToCT'21] and Gutin and Mnich [arXiv:2207.12278].
\mic{Our hardness results apply to the path versions of the problems as well.}
On the positive side, we show that \mic{{\sc Longest Path} parameterized multiplicatively above girth} belongs to the class XP.
\end{abstract}
\clearpage

\section{Introduction}

\subparagraph*{Hamiltonian Cycle (Path) Parameterized by DFVS.}
In {\sc Hamiltonian Cycle (Path)}, we are given a (directed or undirected) graph~$G=(V,E)$, and the objective is to determine whether $G$ contains a simple cycle (path) of length~$n=|V(G)|$, called a {\em Hamiltonian cycle (path)}. {\sc Hamiltonian Cycle}\footnote{Throughout the following paragraphs, we refer only to the cycle variant of the problem, but the same statements also hold for the path variant.} has been widely studied, from various algorithmic~(see, e.g., \cite{DBLP:journals/talg/FominGLSZ19,bergougnoux2020optimal,cygan2018fast,cygan2022solving,DBLP:conf/stacs/Bjorklund21,DBLP:conf/icalp/Bjorklund019,DBLP:conf/icalp/BjorklundKK17,DBLP:journals/siamcomp/Bjorklund14,DBLP:conf/isaac/Bjorklund18,DBLP:conf/focs/BjorklundH13}) and structural~(see, e.g., the survey \cite{gould2003advances}) points of view. This problem is among the first problems known to be NP-complete~\cite{DBLP:conf/coco/Karp72}, and it remains NP-complete on various restricted graph classes (see, e.g.,~\cite{garey1974some,akiyama1980np,buro2001simple}).
Nevertheless, \mic{a longest path can be found easily}
in polynomial-time on the large class of directed acyclic graphs (DAGs) using dynamic programming.\footnote{For example, see the lecture notes at people.csail.mit.edu/virgi/6.s078/lecture17.pdf.} Thus, it is natural to ask---can we 
\mic{solve {\sc Hamiltonian Cycle} efficiently on wider classes of directed graphs that resemble DAGs  to some extent?}

\mic{In the undirected realm, the class of acyclic graphs (i.e., forests) can be generalized to graphs of bounded treewidth~\cite{DBLP:books/sp/CyganFKLMPPS15}.
The celebrated Courcelle's theorem~\cite{CourcelleE12} states that a wide range of problems, including {\sc Hamiltonian Cycle}, are {\em fixed-parameter tractable} (FPT)\footnote{A problem is FPT (resp. XP) if it is solvable in time $f(k)\cdot n^{O(1)}$ (resp. $n^{f(k)}$) for some computable function $f$ of the parameter~$k$, where $n$ is the size of the input.} 
when parameterized by treewidth.}
Historically, the notion of {\em directed treewidth} (Definition \ref{def:directedTw}) was introduced by Johnson et al.~\cite{DBLP:journals/jct/JohnsonRST01} as a generalization of the undirected notion contrived to solving linkage problems such as {\sc Hamiltonian Cycle}; see also~\cite{DBLP:journals/endm/Reed99}. Since then, directed treewidth has been intensively studied~(see, e.g., \cite{kawarabayashi2015directed,adler2007directed,DBLP:journals/corr/abs-1910-01826,DBLP:conf/soda/GiannopoulouKKK22,DBLP:conf/soda/GiannopoulouKKK20,DBLP:conf/soda/HatzelKK19}). Over the years, various other width measures for directed graphs have been proposed and studied as well, where the most prominent ones are the DFVS number (which is the minimum number of vertices to remove to make the graph acyclic), {\em directed pathwidth} (defined similarly to directed treewidth, where we replace trees by paths), {\em DAG-width}~\cite{DBLP:conf/soda/Obdrzalek06} and {\em Kelly-width}~\cite{DBLP:journals/tcs/HunterK08}. Notably, DFVS is the \mic{largest} among them (see Fig.~\ref{fig:ecology}). Johnson et al.~\cite{DBLP:journals/jct/JohnsonRST01} and Hunter et al.~\cite{DBLP:journals/tcs/HunterK08} proved that {\sc Hamiltonian Cycle} (as well as {\sc Longest Cycle}, defined later) is {\em slice-wise polynomial} (XP) parameterized by directed treewidth. Later, Lampis et al.~\cite{DBLP:journals/disopt/LampisKM11} (originally in 2008~\cite{DBLP:conf/isaac/LampisKM08}) showed that, here, this is the ``right'' form of the time complexity: they proved that {\sc Hamiltonian Cycle} parameterized by directed pathwidth is W[1]-hard (in fact, even W[2]-hard), implying that the problem is unlikely to be FPT. 

Thus, since 2008, it has remained open whether {\sc Hamiltonian Cycle} parameterized by DFVS is FPT or W[1]-hard (see Fig.~\ref{fig:ecology}). As a follow-up to their result, Kaouri et al.~\cite{FPTnews} wrote: {\em
``Treewidth occupies a sweet spot in the map
of width parameters: restrictive enough to be efficient
and general enough to be useful. What is the right analogue which occupies a similar sweet spot (if it exists) for digraphs? One direction is to keep searching for the right width that generalizes DAGs, that is searching the area around (and probably above) DFVS. ... The
other possibility is that widths which generalize DAGs,
such as DFVS and all the currently known widths, may
not necessarily be the right path to follow. ... The
search~is~on!''}

\begin{figure}[t]
    \centering
    \includegraphics[scale=0.43]{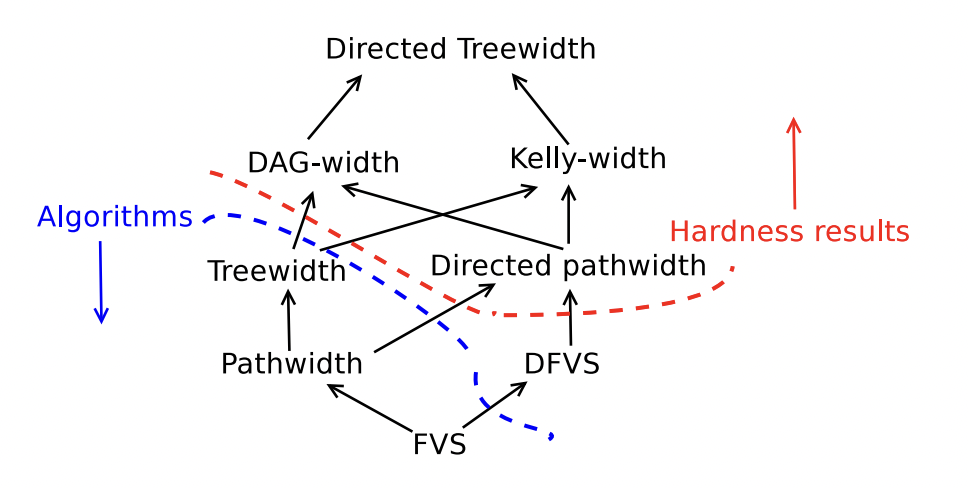}
    \caption{Figure 3 in~\cite{FPTnews}. Caption (verbatim): {\em ``The ecology of digraph widths. Undirected measures refer to the underlying undirected graph. Arrows denote generalizations (for example small treewidth
implies small kelly-width). The dashed lines indicate
rough borders of known tractability and intractability for
most studied problems (including {\sc Hamiltonian Cycle}).''}}
    \label{fig:ecology}
\end{figure}

Our main contribution is, essentially, the finish line for this search for {\sc Hamiltonian Cycle}. We prove that, unfortunately, already for DFVS \mic{number} itself, the problem is W[1]-hard. Formally, {\sc Hamiltonian Cycle (Path) By DFVS} is defined~as~follows: Given a directed graph $G = (V, E)$ and a subset $X \subseteq V(G)$ such that $G - X$ is acyclic, determine whether $G$ contains a Hamiltonian cycle (path). Here, the parameter is $|X|$.
%
%
Note that $X$ need not be given as input---the computation of a minimum-size vertex set whose removal makes a given directed graph acyclic is in FPT~\cite{DBLP:journals/jacm/ChenLLOR08}, so, given a directed graph $G$, we can simply run the corresponding algorithm to attain $X$ whose size is the DFVS number of $G$. 

 \begin{restatable}{theorem}{thmHamCycleDFVS}
 \label{thm:HamCycleDFVS}
{\sc Hamiltonian Cycle By DFVS} is {\em W[1]}-hard.
 \end{restatable}


\mic{We also obtain the same hardness result for {\sc Hamiltonian Path By DFVS}. }
Notice that {\sc Hamiltonian Cycle (Path) By DFVS} is in XP, since, as mentioned earlier, it is already known to be in XP even when parameterized by the smaller parameter directed treewidth~\cite{DBLP:journals/jct/JohnsonRST01}. So, the classification of this problem is resolved.

Lastly, we remark that the choice of the larger directed feedback arc set (DFAS) number is also futile, as it also yields W[1]-hardness. Indeed, there is a simple reduction that shows this: Replace each vertex $v$ in $S$ by two vertices, $v_{\mathsf{in}}$ and $v_{\mathsf{out}}$, so that all vertices that were in-neighbors (out-neighbors) of $v$ become in-neighbors (out-neighbors) of $v_{\mathsf{in}}$ ($v_{\mathsf{out}}$), and add the arc $(v_{\mathsf{in}},v_{\mathsf{out}})$. It is immediate that the original instance is a Yes-instance of {\sc Hamiltonian Cycle} if and only if the new instance is, and that the DFVS number of the original graph equals the DFAS number of the new graph.

\subparagraph*{Longest Cycle (Path) Above Girth.} In the {\sc Longest Cycle (Path)} problem, also known as {\sc $k$-Long Cycle} ({\sc $k$-Path}), we are given a (directed or undirected) graph~$G$ and a non-negative integer~$k$, and the objective is to determine whether~$G$ contains a simple cycle (path) of length at least~$k$. Clearly, {\sc Longest Cycle} is a generalization of {\sc Hamiltonian Cycle} as the latter is the special case of the former when $k=n$. Over the past four decades, {\sc Longest Cycle} and {\sc Longest Path} have been intensively studied from the viewpoint of parameterized complexity. There has been a long race to achieve the best known running time on general (directed and undirected) graphs for the parameter is $k$~\cite{DBLP:journals/jcss/BjorklundHKK17,DBLP:journals/ipl/Williams09,DBLP:conf/icalp/Koutis08,DBLP:conf/esa/Zehavi15,DBLP:journals/jcss/ShachnaiZ16,DBLP:journals/ipl/FominLPSZ18,DBLP:journals/ipl/Zehavi16,DBLP:journals/tcs/Tsur19b,DBLP:journals/jacm/AlonYZ95,DBLP:conf/icalp/BrandP21,monien1985find,chen2009randomized,DBLP:journals/algorithmica/HuffnerWZ08,bodlaender1993linear,gabow2008finding}, where the current winners for directed graphs have time complexity $4^k\cdot n^{O(1)}$~\cite{DBLP:journals/ipl/Zehavi16} for {\sc Longest Cycle} and $2^k\cdot n^{O(1)}$~\cite{DBLP:journals/ipl/Williams09} for {\sc Longest Path}.  Moreover, the parameterized complexity of the problems was analyzed with respect to structural parameterizations (see, e.g.,~\cite{cygan2018fast,cygan2022solving,ganian2015improving,DBLP:journals/talg/FominGLSZ19,bergougnoux2020optimal,DBLP:conf/latin/Golovach0S0Z20,doucha2012cluster}), on special graph classes (see, e.g.,~\cite{dorn2012catalan,lokshtanov2011planar,DBLP:journals/jocg/ZehaviFLP021,DBLP:journals/dcg/FominLPSZ19,DBLP:conf/icalp/FominLP0Z19,DBLP:conf/stoc/Nederlof20a}), and when counting replaces decision~\cite{DBLP:journals/talg/LokshtanovBSZ21,DBLP:journals/talg/AlonG10,DBLP:conf/iwpec/AlonG09,arvind2002approximation,flum2004parameterized,bjorklund2009counting,vassilevska2009finding,
curticapean2017homomorphisms,curticapean2014complexity,
brand2018extensor}. In fact, {\sc Longest Path} is widely considered to be one of the most well studied in the field,  being, perhaps, second only to {\sc Vertex Cover}~\cite{DBLP:books/sp/CyganFKLMPPS15}.

We  consider the ``multiplicative above guarantee'' parameterization for {\sc Longest Path} where the guarantee is girth\footnote{The girth of a graph is the length of its shortest cycle, and it is easily computable in polynomial time.}, called {\sc Longest Cycle (Path) Above Girth}, which is defined as follows (for directed graphs): Given a directed graph~$G = (V, E)$ with girth~$g$ and a positive integer~$k$, determine whether $G$ contains a cycle (path) of length at least~$g \cdot k$. Here, the parameter is $k$.
%
%
This parameterization of {\sc Longest Cycle (Path)} was considered by Fomin et al.~\cite{DBLP:journals/toct/FominGLPSZ21}, who proved that on undirected graphs the problem is in FPT. They posed the following open question: {\em ``For problems on directed graphs, parameterization multiplicatively above girth (which is now the length of a shortest directed cycle) may also make sense. For example, what is the parameterized complexity of {\sc Directed Long Cycle} under this parameterization?''} This question was also posed by Gutin et al.~\cite{gutin2022survey} as Open Question 9. As our second contribution, we resolve this question---in sharp contrast to the undirected case, the answer is negative:

\begin{theorem}\label{thm:LongCycleGirth}
{\sc Longest Cycle Above Girth} is {\em W[1]}-hard.
\end{theorem}

\mic{Again, we can transfer the hardness also to the longest path setting.}
On the positive side, we 
\mic{give a deterministic XP algorithm for the path variant of the problem.}

 \begin{restatable}{theorem}{thmLongPathXP}
 \label{thm:LongPathXP}
{\sc Longest Path Above Girth} is in {\em XP}.
\end{restatable}

\mic{Therefore,} the classification of this problem is resolved \mic{as well}.
Notice that this theorem also easily implies that when we ask for a path of length at least $g+k$, the problem is in FPT 
\mic{(in fact,} solvable in time $2^{O(k)}\cdot n^{O(1)}$) as follows. If $k\geq g$, then we can simply run some known $2^{O(k)}\cdot n^{O(1)}$-time algorithm for {\sc Longest Path} parameterized by the solution size. Otherwise, when $k<g$, we create $n$ instances of {\sc Longest Path Above Girth} with parameter $2$ (thus, solvable in polynomial time by our theorem), one for each vertex~$v$: Replace $v$ by a path of length $g-k$ whose start vertex has all in-neighbors of $v$ as its in-neighbors and whose end vertex has all out-neighbors of $v$ as its out-neighbors, and take the disjoint union of the resulting graph and a cycle of length $g$; it is easy to see that we should return Yes if and only if the answer to at least one of these $n$ instances is Yes.

Other above/below guarantee versions of {\sc Longest Cycle (Path)} were also considered in the literature; see~\cite{DBLP:journals/siamdm/BezakovaCDF19,DBLP:conf/esa/FominGSS22,gutin2022survey,DBLP:conf/wg/Jansen0N19,DBLP:conf/soda/FominGSS22,DBLP:conf/stacs/FominGLSS022,DBLP:conf/sosa/HatzelMPS23,DBLP:journals/corr/abs-2301-06105}. Remarkably, while the parameterized complexity of all of these versions is quite well understood for undirected graphs, the resolution of the parameterized complexity of most of these versions has been (sometimes repeatedly) asked as  an open question for directed graphs, where only little is known. It is conceivable that our contributions will shed light on these open questions as well in future works. 

\subsection{Techniques}
\subparagraph*{Hamiltonian Cycle By DFVS.}
Before describing the ideas behind our reduction, 
we provide some background about the related {\sc Disjoint Paths} problem.
Here, we are given a graph with $k$ vertex pairs $(s_i,t_i)$ and the goal is to determine whether there exists $k$ vertex-disjoint\footnote{Another studied variant involves finding edge-disjoint paths. While these two problems behave differently in some settings, 
this is not the case in our context.} paths such that the $i$-th path connects $s_i$ to $t_i$.  
While {\sc Disjoint Paths} is \mic{famously} known to be FPT  (parameterized by $k$) on undirected graphs~\cite{robertson1995graph}, the problem becomes NP-hard on directed graphs already for $k=2$~\cite{FORTUNE1980111}.
What is interesting for us is that when the input is restricted to be a DAG, then {\sc Disjoint Paths} can be solved by an XP algorithm~\cite{FORTUNE1980111}, but it is unlikely to admit an FPT algorithm as the problem is W[1]-hard~\cite{Slivkins10}.

The latter result suggests a starting point for extending the hardness to our problem.
A~simple idea to design a reduction from {\sc Disjoint Paths} on DAGs 
is to just insert edges $t_1 \to s_2, t_2 \to s_3, \dots, t_k \to s_1$.
Then the set $\{s_1,\dots,s_k\}$ becomes a~DFVS and the existence of $k$ disjoint $(s_i,t_i)$-paths implies the existence of a long cycle.
There is a catch, though.
This construction might turn a No-instance into an instance with a long cycle because the cycle might traverse the  vertices $s_i, t_i$ in a different order, corresponding to a family of disjoint paths which does not form a solution to the original instance.
To circumvent this, we open the black box and
give a reduction directly from the basic W[1]-complete problem---{\sc Multicolored Clique}---while adapting some ideas from the hardness proof for 
 {\sc Disjoint Paths} on DAGs by Slivkins~\cite{Slivkins10}.

 The construction by Slivkins uses two kinds of gadgets.
 First, for each $i = 1, \dots, k$ one needs a {\em choice gadget} comprising two long directed paths with some arcs from the first path to the second one.
 The solution path corresponding to the $i$-th choice gadget must choose one of these arcs to ``change the lane''; the location of this change encodes the choice of the $i$-th vertex in the clique.
 Next, for each pair $(i,j)$, where $1 \le i < j \le k$, one needs a {\em verification gadget} to check whether the $i$-th and $j$-th chosen vertices are adjacent.
 The corresponding solution path must traverse the $i$-th and $j$-th choice gadgets around the locations of their transitions.
 The arcs between the choice gadgets are placed in such a way that both these locations can be visited only when the corresponding edge exists in the original graph. 

 Our reduction relies on a similar mechanism of choice gadgets, formed by $k$ directed paths, \mic{corresponding to $k$ colors in the {\sc Multicolored Clique} instance.}
 \mic{The $i$-th path is divided into blocks corresponding to the vertices of color $i$.}
 We enforce that a Hamiltonian cycle must omit exactly one block when following such a path, thus encoding a choice of~$k$ vertices.
 The arcs between blocks from different paths encode the adjacency matrix of the original graph. 
 We would like to guarantee that any Hamiltonian cycle $C$ visits each omitted block $k-1$ times, utilizing the arcs mentioned above.
 To ensure this, we need to allow $C$ to ``make a step backward'' on the directed path during such a visit.
 On the other hand, we cannot afford to create too many disjoint cycles because DFVS should have size bounded in terms of~$k$.
As a remedy, we attach $k-1$ long cycles to each choice gadget, which are connected to every block in a certain way (see \Cref{fig:ham-dfvs-1} on page \pageref{fig:ham-dfvs-1}).
Using the fact that every vertex in the gadget must be visited by $C$ at some point,
we prove that each long cycle is entered exactly once and each choice gadget is entered exactly $k$ times in total, imposing a rigid structure on a solution.
As a result, we get rid of the flaw occurring in the naive construction
and obtain an equivalent instance with DFVS number $\mathcal{O}(k^2)$.

\subparagraph*{Longest Cycle (Path) Above Girth.}
Having established our main hardness result, it is easy to extend it to {\sc Longest Cycle (Path) Above Girth}.
We replace each vertex $v$ in the DFVS by two vertices $v_{in}$ and $v_{out}$, splitting the in- and out-going arcs of $v$ between them, and insert a long path from $v_{in}$ to $v_{out}$.
This path is set sufficiently long to make the girth $g$ of the graph comparable to its size.
Consequently, the original graph $G$ is Hamiltonian if and only if the second one has a cycle of length $g \cdot (\mathsf{dfvs}(G) + 1)$.
So, this problem is~also~W[1]-hard.

To design an XP algorithm for {\sc Longest Path Above Girth}, we follow the win-win approach employed by the FPT algorithm for the undirected case~\cite{DBLP:journals/toct/FominGLPSZ21}, which relies on a certain version of the Grid Minor Theorem.
In general, the theorem states that either a graph contains a $k \times k$-grid as a minor, or its treewidth is bounded by $f(k)$, for some (in fact, polynomial) function $f$.
In the first scenario, a sufficiently long path always exists, whereas in the second scenario, the problem is solved via dynamic programming on the~tree~decomposition.

We take advantage of the directed counterpart of the Grid Minor Theorem: either a~digraph contains a subgraph isomorphic to a subdivision of the order-$\OO(k)$ cylindrical wall or its directed treewidth is bounded by $f(k)$, for some function $f$~\cite{kawarabayashi2015directed}.
We prove that in the first scenario again a sufficiently long path always exists (so we obtain a Yes-instance), while in the second scenario we can utilize 
\mic{the known algorithm to compute the longest path in XP time with the help of the directed tree decomposition.
Finally, let us remark that this argument does not extend to {\sc Longest Cycle Above Girth} because we cannot guarantee the existence of a sufficiently long cycle in a subdivision of a cylindrical wall.
We leave it open whether this problem belongs to XP as well.}

\section{Preliminaries}\label{section:prelims}
\subparagraph*{General Notation.}
For an integer $r$, let $[r] = \{1,\ldots,r\}$, and for integers $r_1, r_2$,  let $[r_1,r_2] = \{r_1,r_1+1, \ldots,r_2\}$.
We refer to \cite{DBLP:books/sp/CyganFKLMPPS15} for standard definitions in Parameterized Complexity. 
\subparagraph*{Directed Graphs.}
We use standard graph theoretic terminology from Diestel's book~\cite{diestel-book}. 
We work with simple directed graphs (digraphs) where the edges are given by a set $E$ of ordered pairs of vertices.
For an edge $e = (u, v)$ in a digraph $G$, we say that $u$ is an {\em in-neighbor} of $v$ and $v$ is an {\em out-neighbor} of $u$. We refer to $e$ as an {\em incoming edge of $v$} and an {\em outgoing} edge of $u$. 
%
For a vertex set $S \subseteq V(G)$ we define  $\partial^{out}(S) = \{(u,v) \in E(G) \colon u \in S\}$, $\partial^{in}(S) = \{(u,v) \in E(G) \colon v \in S\}$, and $\partial(S) = \partial^{out}(S) \cup \partial^{in}(S)$.
For two vertex sets $A, B \subseteq V(G)$ we write $E(A,B)$ for $\partial(A) \cap \partial(B)$, that is, the set of edges with one endpoint in $A$ and the other in $B$.
We use this notation only when the digraph $G$ is clear from the context.
When $C$ is a subgraph of $G$ we abbreviate $\partial(C) = \partial(V(C))$ and likewise for the remaining notation.
A digraph $G$ is {\em isomorphic} to a digraph $H$ if there is a bijection $f : V ( G ) \rightarrow V ( H )$ such that for any $u,v \in V(G)$ we have $(u,v) \in E(G)$ if and only if $(f(u), f(v)) \in E(H)$.
%
The {\em length} of a path (or a cycle) $P$  is the number of 
edges in $P$.
When the first and last vertices of a path $P$ are $s$ and $t$, respectively, we call it an {\em $(s,t)$-path}.
A cycle in $G$ is called {\em Hamiltonian} if it visits all the vertices of $G$.
The {\em girth} of a digraph $G$ is the shortest length of a cycle in $G$.  
%
%
For a rooted tree $T$ and a node $t \in V(T)$, $T_t$ denotes the subtree of $T$ rooted at $t$. 
By orienting each edge in a rooted tree from a parent to its child we obtain an {\em arborescence}.

\subparagraph*{Directed Treewidth.}
We move on to the directed counterparts of treewidth and grids.
\begin{definition}[Directed Tree Decomposition~\cite{kawarabayashi2015directed}]\label{def:directedTw}
A {\em directed tree decomposition} of a directed graph $G$ is a triple $(T,\beta,\gamma)$, where $T$ is an arborescence, and $\beta : V(T) \rightarrow 2^{V(G)}$ and $\gamma : E(T) \rightarrow 2^{V(G)}$ are functions such that
\begin{enumerate}
\item $\{\beta(t) : t \in V (T)\}$ is a partition of $V(G)$ into (possibly empty) sets, and
\item  if $e=(s,t) \in E(T)$, $A = \bigcup \{\beta(t') : t' \in V(T_t)\}$ and $B = V(G) \setminus A$, then there is no closed (directed) walk in $G - \gamma(e)$ containing a vertex in $A$ and a vertex in $B$.
\end{enumerate}

For $t \in V(T)$, define $\Gamma(t) := \beta(t) \cup \bigcup \{\gamma(e) : $ $e$ is incident with $t\}$. Moreover, define $\beta(T ) := \bigcup \{\beta(t') : t' \in V(T )\}$.

The {\em width} of $(T, \beta, \gamma)$ is the least integer $w$ such that $|\Gamma(t)| \leq w + 1$ for all $t \in V(T )$. The {\em directed treewidth} of $G$ is the least integer $w$ such that $G$ has a directed tree decomposition of width $w$.
\end{definition}


\begin{definition}[Subdivision]
For a directed graph $G$, the {\em edge subdivision} of  $(u,v) \in E(G)$ is the operation that removes $(u,v)$ from $G$ and  inserts two edges $(u,w)$ and $(w,v)$ with the new vertex $w$. A graph derived from $G$ by a sequence of edge subdivisions is a {\em subdivision}~of~$G$.
\end{definition}

We now define special graphs called cylindrical grids and cylindrical walls (see Fig.~\ref{fig:cylindrical-grid-and-wall}).

\begin{figure}
\begin{center}
\includegraphics[scale=0.45]{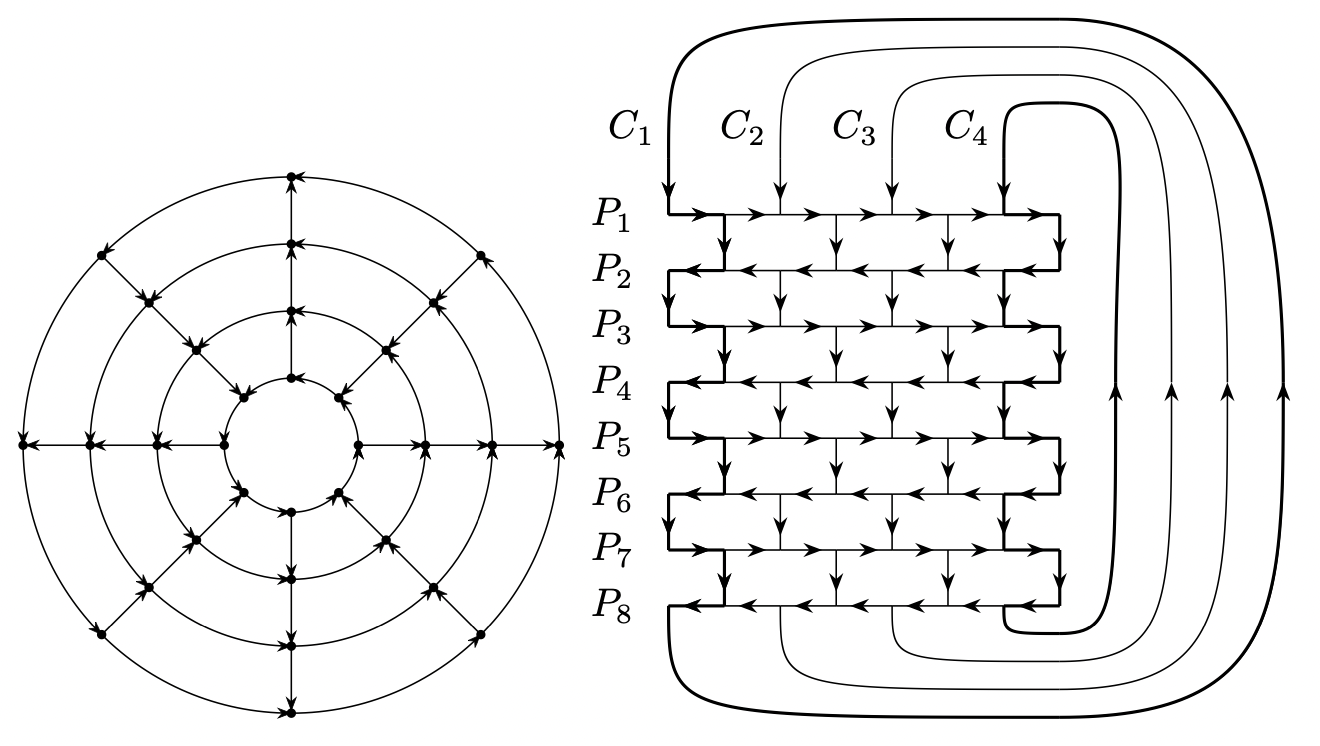}  
\end{center}
\caption{A cylindrical grid and a cylindrical wall of order $4$. The figure is taken from \cite{kawarabayashi2015directed}.}\label{fig:cylindrical-grid-and-wall}
\end{figure}

\begin{definition}[Cylindrical Grid and Cylindrical Wall~\cite{campos2019adapting}]\label{definition:cylindrical-grid-wall}
A {\em cylindrical grid of order $k$}, for some $k \geq 1$, is the directed graph $G_k$ consisting of $k$  pairwise vertex disjoint directed cycles $C_1, \ldots, C_k$, together with $2k$ pairwise vertex disjoint directed paths $P_1 \ldots, P_{2k}$ such that: 
\begin{itemize}
\item for $i \in [k]$, $V (C_i) = \{v_{i,1}, v_{i,2}, \dotsc , v_{i,2k}\}$ and $E(C_i) = \{(v_{i,j} , v_{i,j+1}) |\ j \in [2k - 1]\} \cup \{(v_{i,2k}, v_{i,1})\}$
\item for $i \in \{1, 3, 5, \dotsc , 2k - 1\}$, $E(P_i) = \{(v_{1,i}, v_{2,i}), (v_{2,i}, v_{3,i}), \dotsc , (v_{k-1,i}, v_{k,i})\}$,
and
\item for $i \in \{2, 4, 6, \dotsc , 2k\}$, $E(P_i) = \{(v_{k,i}, v_{k-1,i}), (v_{k-1,i}, v_{k-2,i}), \dotsc , (v_{2,i}, v_{1,i})\}$.

\end{itemize} 
A {\em cylindrical wall of order $k$} is the directed graph $W_k$ obtained from the cylindrical grid $G_k$ by splitting each vertex of degree $4$ as follows: we replace $v$ by two new vertices $v_{in}, v_{out}$ and an edge $(v_{in}, v_{out})$ so that every edge $(u,v) \in E(G_k)$ is replaced by an edge $(u,v_{in})$ and every edge $(v, u) \in E(G_k)$ is replaced by an edge $(v_{out}, u)$.
\end{definition}

Note that in the cylindrical grid $G_k$, the path $P_i$ is oriented from the first cycle to the last one if $i$ is odd, and from the last cycle to the first if $i$ is even (see Figure \ref{fig:cylindrical-grid-and-wall}). We have the following theorem for directed graphs, which will be helpful in designing our algorithm in Section \ref{section:fpt-long-path-above-girth} using a win/win approach. 

\begin{theorem}[Directed Grid Theorem~\cite{kawarabayashi2015directed}]\label{theorem:directed-grid-theroem}
There is a function $f: \mathbb{N} \rightarrow \mathbb{N}$ such that 
for every fixed $k \in \mathbb{N}$, when given a directed graph $G$, in polynomial time we can compute either:
\begin{enumerate}
\item 
a subgraph of $G$ that is isomorphic to a subdivision\footnote{The theorem in~\cite{kawarabayashi2015directed} states that the cylindrical wall of order $k$ is obtained as a topological minor of $G$. For any topological minor $H$ of $G$, there exists a subdivision of $H$ isomorphic to a subgraph of $G$.} of $W_k$ or
\item a directed tree decomposition of $G$ of width at most $f(k)$.
\end{enumerate}
\end{theorem}


\section{Hardness of {\HamDFVS}}

This section is devoted to the proof of \Cref{thm:HamCycleDFVS}. 
It is based on a parameterized reduction from
{\sc Multicolored Clique}, defined as follows.

\defparproblem{{\sc Multicolored Clique}}{A graph $G = (V, E)$, an integer $k$, and a partition  $(V^1, V^2, \dotsc ,V^k)$ of $V$.}{$k$}{Is there a clique of size $k$ with a vertex from each $V^i$, $i \in [k]$?}

This problem is well-known to be W[1]-hard~\cite[Theorem 13.7]{DBLP:books/sp/CyganFKLMPPS15}. For an instance $(G, (V^1, V^2, \dotsc ,V^k))$ of {\sc Multicolored Clique}, we construct an instance $(G',X)$ of {\HamDFVS} as follows.

\subparagraph*{Construction of $G'$.} For $i \in [k]$ we construct a directed path $P^i$ corresponding to $V^i$ as follows. Let us fix an arbitrary ordering $<_i$ of the vertices in $V^i$, and accordingly, denote $V^i = \{v_1, v_2, \dotsc v_{|V_i|}\}$. To each vertex $v \in V^i$
we associate a directed path $P_{v}$ on $2k$ vertices. We let $v^{\lef}$ and $v^{\ri}$ denote the first and last vertices of $P_{v}$, respectively. 
We refer to the $2(k-1)$ internal vertices of $P_{v}$ as $v^{1,\out}, v^{1,\inn}, v^{2,\out}, \dotsc, v^{i-1,\out}$, $v^{i-1,\inn}, v^{i+1,\out}, v^{i+1,\inn}, \dotsc, v^{k,\out}, v^{k,\inn}$
(note that the index $i$ is avoided).  
The directed path $P^i$ is the concatenation of these paths: $P_{v_1} \rightarrow P_{v_2} \rightarrow, \dotsc, \rightarrow P_{v_{|V^i|}}$. 

For every $i \in [k]$
 we create a ``universal'' vertex $u^i$ and insert edges $(u_i, v^\lef)$, $(u_i, v^\ri)$, $(v^\lef, u_i)$, $(v^\ri, u^i)$ for all $v \in V^i$.
Next, we create $k-1$ cycles $C^{i \rightarrow j}$  
for $j \in [k] \setminus \{i\}$, each of length $2\cdot|V^i|$.
The vertices of $C^{i\rightarrow j}$ are $c^{i\rightarrow j,\out}_{v_1}, c^{i\rightarrow j,\inn}_{v_1},  c^{i\rightarrow j,\out}_{v_2},  c^{i\rightarrow j,\inn}_{v_2},  \dotsc, c^{i\rightarrow j,\out}_{v_{|V_i|}}, c^{i\rightarrow j,\inn}_{v_{|V_i|}}$.
We insert edges from $c^{i\rightarrow j,\out}_{v}$ to $v^{j,\out}$ and  from $v^{j,\inn}$ to $c^{i\rightarrow j,\inn}_{v}$ for all $v \in V^i$. See Figure \ref{fig:ham-dfvs-1} for an illustration.

\begin{figure}
\begin{center}
\includegraphics[scale=0.4]{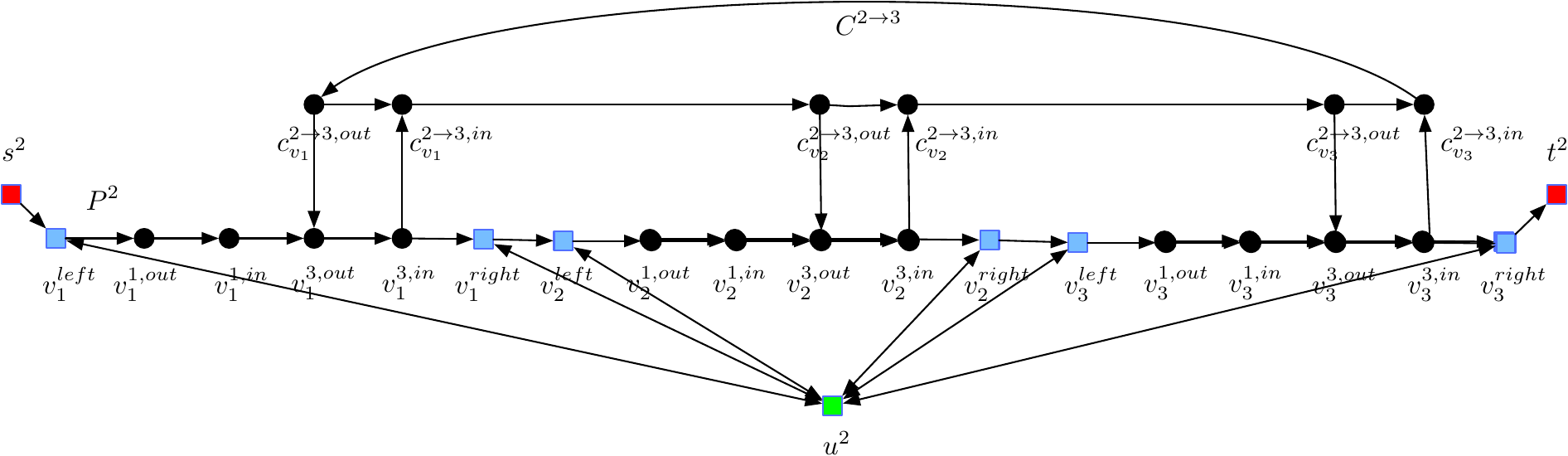}  
\end{center}
\caption{Fragment of the graph $G'$: the path $P^2$ 
(comprising subpaths $P_{v_1}, P_{v_2}, P_{v_3}$)
and the cycle $C^{2\rightarrow 3}$. 
}\label{fig:ham-dfvs-1}
\end{figure}

We add two sets of $k$ ``terminal'' vertices: $S = \{s^1, s^2, \dotsc, s^k\}$ and $T = \{t^1, t^2, \dotsc, t^k\}$. 
For $i \in [k]$
we insert  edges from $s^i, i \in [k]$, to $v^{\lef}_1$ (being the first vertex of the path $P^i$). 
We also add edges from $v^{\ri}_{|V^i|}$ (being the last vertex of the path $P^i$) 
to $t^i$. 

We also create two additional sets of $\binom{k}{2}$ terminal vertices each:
$\widehat{S} = \{\widehat{s}^{i \rightarrow j} \colon 1 \le 1 < j \le k\}$
and $\widehat{T} = \{\widehat{t}^{i \rightarrow j} \colon 1 \le 1 < j \le k\}$.
For $i, j \in [k], i <j$, and $v \in V^i$, we insert an edge from $\widehat{s}^{i \rightarrow j}$ to $v^{j, \inn}$  (which belongs to $P_{v}$ and hence to $P^i$) and for $v \in V^j$, we insert an edge from $v^{i,\out}$ 
to $\widehat{t}^{i \rightarrow j}$. 
Furthermore, we add edges from every  $ t \in T \cup \widehat T$ to every $s \in S \cup \widehat S$.

Finally we encode the adjacency matrix of $G$: for each edge $uv$ in $E(G)$ with $u \in V^i$ and $v \in V^j$ with $i<j$, we insert an edge from $u^{j,\out}$ to $v^{i,\inn}$. 
See Figure \ref{fig:ham-dfvs-3} for an example.

\begin{figure}[t]
\includegraphics[scale=0.4]{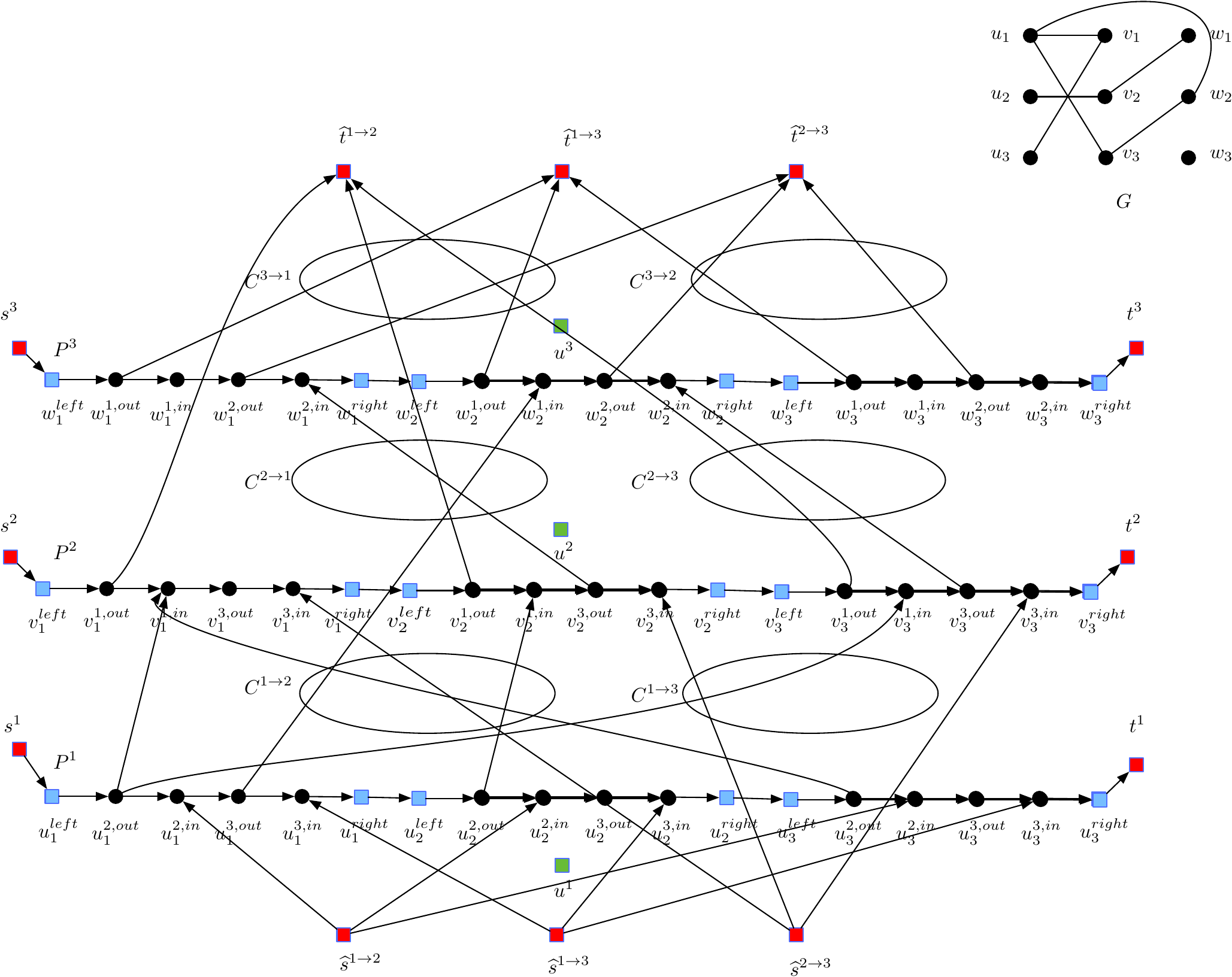} 
\caption{Graph $G$ (top right) with a clique $\{u_1,v_3,w_2\}$ and the graph $G'$ encoding $G$.
For the sake of legibility, several groups of edges are omitted in this picture:
edges incident to $u^i$ or the cycles $C^{i \rightarrow j}$ and edges going from $T \cup \widehat T$ to $S \cup \widehat S$.}
\label{fig:ham-dfvs-3}
\end{figure}

This concludes the construction of $G'$.
Clearly, the graph $G'$ can be computed in polynomial time when given $(G,(V^1,\dots,V^k))$.
We begin the analysis of $G'$
by showing that it has a DFVS of size $\OO(k^2)$.

\begin{lemma}\label{lemma:DFVS-k^2}
There exists a subset $X \subseteq V(G')$ such that $G'-X$ is a DAG and $|X| = k(k-1)+2k+\binom{k}{2}$.    
\end{lemma}
\begin{proof}
Let $Y$ be the set of vertices 
$\{c^{i \rightarrow j,\out}_{v^i} \colon i,j \in [k], \, i\ne j\}$
where $v^i$ stands for the first vertex in $V^i$.
We set $X = Y \cup T \cup \widehat T \cup \{u^1, \dotsc u^k\}$ and claim that $G'-X$ is a directed acyclic graph. 
First observe that for each $i \in [k]$
the graph $L^i := P^i \cup \bigcup_{j\ne i} C^{i \rightarrow j} - X$ is acyclic because it can be drawn with each edge being either vertical or horizontal facing right (see \Cref{fig:ham-dfvs-1}).
The remaining edges in $G' - X$ either start at $ S \cup \widehat S$ (these vertices have no incoming edges) or go from $L^i$ to $L^j$ for $i<j$.
Thus $G' - X$ is a DAG. 
Finally we check that $|X| = |Y| + |T| + |\widehat T | + k = k(k-1)+2k+\binom{k}{2}$.
\end{proof}


\subparagraph*{Correctness.}
We first show that if $(G, (V^1, V^2, \dotsc ,V^k))$ is a Yes-instance then there exists a Hamiltonian cycle in $G'$.
In the following lemmas, we refer to the directed feedback vertex set $X$ from Lemma \ref{lemma:DFVS-k^2}.

\begin{lemma}\label{lemma:if-clique-then-hampath} 
If $(G, (V^1, V^2, \dotsc ,V^k))$ is a Yes-instance of {\sc Multicolored Clique} then  $(G',X)$ is a Yes-instance of {\HamDFVS}.
\end{lemma}
\begin{proof}
Let $\{v_1, \dotsc, v_k \}$ be a clique in $G$ with $v_i \in V^i$ for $i \in k$. 
For $i \in [k]$ we define the path $Z^i$ that starts at $s^i$, follows $P^i$ until $v_i^\lef$, visits $u^i$, and then follows $P^i$ from  $v_i^\ri$ to $t^i$.
Next, for $i<j$ we construct the path  
$Q^{i\rightarrow j}$ as $\widehat{s}^{i\rightarrow j} \rightarrow v_i^{j,\inn} \rightarrow C^{i\rightarrow j} \rightarrow v_i^{j,\out} \rightarrow v_j^{i,\inn} \rightarrow C^{j\rightarrow i} \rightarrow v_j^{i,\out} \rightarrow \widehat{t}^{i\rightarrow j}$.
Note that one can traverse the entire cycle $C^{i\rightarrow j}$ after entering it from $v_i^{j,\inn}$ and leave towards $v_i^{j,\out}$.
The edge $(v_i^{j,\out}, v_j^{i,\inn})$ is present in $G'$ due to the encoding of the adjacency matrix of $G$.
The union of all these paths covers the entire vertex set of $G'$.
It suffices to observe that
these paths can be combined into a single cycle using the edges from $T \cup \widehat T$ to $S \cup \widehat S$.
\end{proof}

 Proving the converse of Lemma \ref{lemma:if-clique-then-hampath} is more challenging.
We will reveal the structure of a potential Hamiltonian cycle $H$ in $G'$ gradually.
First, we show that if $H$ enters
the cycle $C^{i\rightarrow j}$ through an incoming edge corresponding to a certain vertex $v \in V^i$,
then it must also leave it through an outgoing edge related to $v$, and vice versa.
A priori, there might be multiple vertices  $v \in V^i$ for which this happens.
 
 For $v \in V^i$ we denote the edge from $c^{i\rightarrow j, \out}_{v}$ to $v^{j,\out}$ by $e^{i\rightarrow j,\out}_{v}$. Similarly, let the edge from $v^{j,\inn}$ to $c^{i\rightarrow j,\inn}_{v}$ be $e^{i\rightarrow j,\inn}_{v}$. 
 

\begin{lemma}\label{lemma-H-cycle-edges}
Let $H$ be a Hamiltonian cycle in $G'$ and $1\le i < j \le k$. Then $E(H) \cap \partial(C^{i\rightarrow j})$ 
is of the form $\bigcup_{v \in U} \{e^{i\rightarrow j,\out}_{v}, e^{i\rightarrow j,\inn}_{v}\}$ where $U$ is some subset of $V^i$.
\end{lemma}
\begin{proof}
Note that $\partial(C^{i \rightarrow j}) = \{ e^{i\rightarrow j,\out}_{v_1}, e^{i\rightarrow j,\inn}_{v_1}, e^{i\rightarrow j,\out}_{v_2}, e^{i\rightarrow j,\inn}_{v_2}, \dotsc e^{i\rightarrow j,\out}_{v_{|V_i|}}$ $, e^{i\rightarrow j,\inn}_{v_{|V_i|}} \}$.
First suppose that for some $v \in V^i$, we have 
$e^{i\rightarrow j,\inn}_{v} \in E(H)$. Since $H$ is a Hamiltonian cycle in $G'$, it has to visit the vertex $c^{i\rightarrow j,\out}_{v}$. Thus $H$ contains an outgoing edge of $c^{i\rightarrow j,\out}_{v}$. The only two outgoing edges of  $c^{i\rightarrow j,\out}_{v}$  are $e^{i\rightarrow j,\out}_{v}$ and $(c^{i\rightarrow j,\out}_{v}, c^{i\rightarrow j,\inn}_{v})$. Since $e^{i\rightarrow j,\inn}_{v} \in E(H)$ is an incoming edge of $c^{i\rightarrow j,\inn}_{v}$, the edge  $(c^{i\rightarrow j,\out}_{v}, c^{i\rightarrow j,\inn}_{v})$
cannot belong to $E(H)$.
The only other candidate for the outgoing edge of $c^{i\rightarrow j,\out}_{v}$ is $e^{i\rightarrow j,\out}_{v}$, so $e^{i\rightarrow j,\out}_{v} \in E(H)$.

The proof of implication in the other direction is analogous.
Hence for each  $v \in V^i$ we have
$e^{i\rightarrow j,\out}_{v} \in E(H) \Leftrightarrow e^{i\rightarrow j,\inn}_{v} \in E(H)$.
\end{proof}

As the next step, we show that when  $v \in V^i$ and a Hamiltonian cycle enters some cycle $C^{i\rightarrow j}$ adjacent to $P^i$ through a vertex from $P_{v}$, then it must enter all the $k-1$ cycles adjacent to $P^i$ through $P_{v}$.
Again, we cannot yet exclude a scenario that this would happen for multiple vertices  $v \in V^i$.


\begin{lemma}\label{lemma-H-cycle-edges-union}
For any Hamiltonian cycle $H$ in $G'$ and $i \in [k]$,
there exists a  subset $U \subseteq V^i$ for which
we have $E(H) \cap E(P^i, \bigcup_{j \neq i} C^{i\rightarrow j}) =
\bigcup_{v \in U} E(P_{v}, \bigcup_{j \neq i} C^{i\rightarrow j})$.
\end{lemma}
\begin{proof}
Targeting a contradiction, suppose there exists a vertex $v \in V^i$ such that at least one edge of $E(P_{v},$ $ \bigcup_{j \neq i} C^{i\rightarrow j})$ is in $E(H)$ and at least one edge of $E(P_{v},$ $ \bigcup_{j \neq i} C^{i\rightarrow j})$ is not in $E(H)$. 
From Lemma \ref{lemma-H-cycle-edges} we know that $E(H) \cap E(P_{v}, \bigcup_{j \neq i} C^{i\rightarrow j})$ is a union of pairs of the form $\{e^{i\rightarrow j,\out}_{v}, e^{i\rightarrow j,\inn}_{v}\}$ for some $j \in [k] \setminus \{i\}$. 
Hence there indices $p, j \in [k] \setminus \{i\}$ so that $(v^{p,\inn}, v^{j,\out}) \in E(P^i)$, one of the sets $E(B_{v}, C^{i\rightarrow p})$, $E(B_{v}, C^{i\rightarrow j})$ has empty intersection with $E(H)$ and the other one is contained in $E(H)$.
We can assume w.l.o.g. that $j = p + 1$ and the first of these sets is empty.

We arrive at the following scenario: $p \in [k] \setminus \{i\}$, $e^{i\rightarrow p,\out}_{v}, e^{i\rightarrow p,\inn}_{v} \notin E(H)$ and $e^{i\rightarrow p+1,\out}_{v}$ $, e^{i\rightarrow p+1,\inn}_{v} \in E(H)$. Since $H$ is a Hamiltonian cycle, it has to visit the vertex $v^{p,\inn}$. Thus, it has to traverse an outgoing edge of $v^{p,\inn}$. The only outgoing edges of $v^{p,\inn}$ are the edges $e^{i\rightarrow p,\inn}_{v}$ and $(v^{p,\inn}, v^{p+1,\out})$. By the assumption, $e^{i\rightarrow p,\inn}_{v} \notin E(H)$ so $(v^{p,\inn}, v^{p+1,\out}) \in E(H)$. 
On the other hand, since $e^{i\rightarrow p+1,\out}_{v} \in E(H)$ is an incoming edge of $v^{p+1,\out}$, we have $(v^{p,\inn}, v^{p+1,\out}) \notin E(H)$ as it is also an incoming edge to $v^{p+1,\out}$. 
This yields a contradiction.
\end{proof}

We want to argue now that in fact the set $U$ in \Cref{lemma-H-cycle-edges-union} is a singleton, that is, for fixed $i \in [k]$ each cycle $C^{i\rightarrow j}$ 
is being entered exactly once and from the same subpath $P_{v}$ of $P^i$.
To this end, we take advantage of the universal vertices $u^i$.

\begin{lemma} \label{lemma-H-cycle-edges-one-vertex-only}
For any Hamiltonian cycle $H$ in $G'$ and $i\in [k]$ there exists $v \in V^i$ such that $E(H) \cap E(P^i, \bigcup_{j \neq i} C^{i\rightarrow j})$ is $E(P_{v}, \bigcup_{j \neq i} C^{i\rightarrow j})$. 
\end{lemma}
\begin{proof}
Let $U$ be the set from \Cref{lemma-H-cycle-edges-union}.
Targeting a contradiction, suppose that there exist two distinct $v,w \in U$. That is,  $E(H) \cap E(P^i, \bigcup_{j \neq i} C^{i\rightarrow j}) \supseteq E(P_{v}, \bigcup_{j \neq i} C^{i\rightarrow j}) \cup E(P_{w}, \bigcup_{j \neq i} C^{i\rightarrow j})$. 

The Hamiltonian cycle $H$ must contain an outgoing edge of $v^{\lef}$, which is the first vertex on the path $P_{v}$.
The only out-neighbors of $v^{\lef}$ are $u^i$ and  $v^{p,\out}$ where $p=1$ when $i \neq 1$ or $p=2$ when $i=1$.
Recall that the $e^{i\rightarrow p,\out}_{v}$ is present in $E(B^i_{v}, \bigcup_{j \neq i} C^{i\rightarrow j})$ and thereby in $E(H)$ by the assumption. Also, since the edge $e^{i\rightarrow p,\out}_{v}$ is an incoming edge of $v^{p,\out}$ , the edge  $(v^{\lef}, v^{p, \out})$ cannot be in $E(H)$. 
Thus the outgoing edge of $v^{\lef}$ in $H$ is $(v^{\lef}, u^i)$.

Now consider the vertex $w^{\lef}$, which is the first one of the path $P_{w}$. By the same argument as above, we have $(w^{\lef}, u^i) \in E(H)$.
This implies that $u^i$ has two in-neighbors in $H$, a contradiction.

It is also impossible that $U = \emptyset$ because then  $E(H) \cap \partial (C^{i\rightarrow j}) = \emptyset$ for each $j \ne i$ implying that $H$ is disconnected.
Consequently, $U$ contains exactly one element.
%
\end{proof}

We can summarize the arguments given so far as follows:
 for a Hamiltonian cycle $H$ in $G'$ and $i \in [k]$, there exists $v_i \in V^i$, such that, for all $j \neq i$, it holds that $E(H) \cap \partial(C^{i\rightarrow j}) = E(P_{v_i}, C^{i\rightarrow j}) = \{e^{i\rightarrow j, \out}_{v_i}, e^{i\rightarrow j, \inn}_{v_i}\}$.
We make note of a simple implication of this fact.

\begin{lemma}\label{lemma:cycle-out-in}
Let $H$ be a Hamiltonian cycle in $G'$ and $i,j\in [k]$.
Let $v \in V^i$ satisfy  $E(H) \cap \partial(C^{i\rightarrow j}) = \{e^{i\rightarrow j, \out}_{v}, e^{i\rightarrow j, \inn}_{v}\}$.  
Then 
$e=(v^{j,\out}, v^{j,\inn})$ does not belong to $E(H)$.    
\end{lemma}
\begin{proof}
Since $E(H) \cap \partial(C^{i\rightarrow j})$ comprises exactly two edges,
we infer that $H$ contains the path $Q= v^{j,\inn} \rightarrow c^{i \rightarrow j,\inn}_{v} \rightarrow \dotsc \rightarrow c^{i \rightarrow j,\out}_{v} \rightarrow v^{j,\out}$ with all internal vertices from $V(C^{i\rightarrow j})$.
If $H$ traversed the edge $e=(v^{j,\out}, v^{j,\inn})$, then it would contain the cycle $C$ formed by $Q$ and $e$. This would imply $H=C$, which contradicts that $H$ is Hamiltonian.
\end{proof}

We are going to show that 
$H$ can include only one edge that goes from $V(P^i)$ to $V(P^j)$;
it will follow that this edge must be $(v^{j,\out}_i, v^{i,\out}_j)$.

\begin{lemma}\label{lemma-Pij}
Let $H$ be a Hamiltonian cycle in $G'$. 
For each pair $i,j \in [k]$ with $i < j$ we have $|E(H) \cap E(P^i, P^j)| \le 1$.
\end{lemma}
\begin{proof}
Suppose that $|E(H) \cap E(P^i, P^j)| \ge 2$.
Then there exist $u, w \in V^j$ and $e_u \in \partial^\inn(u^{i,\inn}) \cap \partial^{\out}(P^i)$, $e_w \in \partial^\inn(w^{i,\inn}) \cap \partial^{\out}(P^i)$ such that $e_u, e_w \in E(H)$.
This implies that the edges $(u^{i,\out}, u^{i,\inn})$, $(w^{i,\out}, w^{i,\inn}) \in E(P^j)$ cannot be used by $H$.
Since $i < j$, the vertex $u^{i,\out}$ has only one out-neighbor different than $u^{i,\inn}$: the terminal $\widehat{t}^{i \rightarrow j}$.
Hence $(u^{i,\out}, \widehat{t}^{i \rightarrow j}) \in E(H)$.
But the same argument applies to $w^{i,\out}$.
As a consequence, two incoming edges of $\widehat{t}^{i \rightarrow j}$ are being used by $H$ and so we arrive at a contradiction.
\end{proof}

Finally, we prove the second implication in the correctness proof of~the~reduction.

\begin{lemma}\label{lemma:if-hampath-then-clique} 
If  $(G',X)$ is a Yes-instance of {\HamDFVS} then  $(G, (V^1, V^2, \dotsc ,V^k))$ is a Yes-instance of {\sc Multicolored Clique}.
\end{lemma}
\begin{proof}
    Let $H$ be a Hamiltonian cycle in $G'$ and
    $v_i \in V^i$ be the vertex given by \Cref{lemma-H-cycle-edges-one-vertex-only} for  $i \in [k]$.
Fix a pair of indices  $i < j$.
By \Cref{lemma:cycle-out-in}  we know that the edge $(v_i^{j,\out}, v_i^{j,\inn})$ does not belong to $E(H)$. 
The remaining outgoing edges of $v_i^{j,\out}$ belong to $E(P^i,P^j)$ and $H$ must utilize one of them.
By the same argument, 
$H$ must traverse one of the incoming edges of $v_j^{i,\inn}$ that belongs to $E(P^i,P^j)$.
Due to \Cref{lemma-Pij}, the Hamiltonian cycle $H$ can use at most one edge from $E(P^i,P^j)$.
Consequently, we obtain that $(v_i^{j,\out}, v_j^{i,\inn}) \in E(H)$.
In particular, this means that $(v_i^{j,\out}, v_j^{i,\inn})$ is present in $E(G')$ and, by the construction of $G'$, implies that $v_iv_j \in E(G)$.
Therefore, $\{v_1, v_2, \dots, v_k\}$ forms a clique in $G$.
\end{proof}

Lemmas \ref{lemma:if-clique-then-hampath} and \ref{lemma:if-hampath-then-clique} constitute that the instances   $(G, (V^1, V^2, \dotsc ,V^k))$ and $(G',X)$ are equivalent while
\Cref{lemma:DFVS-k^2}
ensures that $|X| = \OO(k^2)$.
We have thus obtained a parameterized reduction from {\sc Multicolored Clique} to \HamDFVS, proving \Cref{thm:HamCycleDFVS}.

\section{Hardness of {\LongCycleG}}
\label{sec:LongPathG}

In this section, we prove \Cref{thm:LongCycleGirth} by giving a parameterized reduction from {\HamDFVS} to {\LongCycleG}. 

\medskip \noindent {\bf Construction.} Consider an instance $(G,X)$ of {\HamDFVS}.
Let $k=|X|$ and $n=|V(G)|$. 
We assume that $k \ge 2$ as otherwise we can solve $(G,X)$ in polynomial time.
Let $x$ be an arbitrary vertex from $X$.
We construct an instance $(G',k+1)$ of {\LongCycleG} as follows. 

Start with $G'=G$. For every $v \in X$, do the following:
\begin{enumerate}
    \item replace $v$ with two vertices $v_{in}$ and $v_{out}$,
    \item add edges $(u, v_{in})$ for every $(u,v) \in E(G)$ and $(v_{out},u)$ for every $(v,u) \in E(G)$,
    \item if $v = x$, then add a directed path of length $n+k-1$ from $v_{in}$ to $v_{out}$ with newly introduced $n+k-2$ internal vertices,
    \item if $v \ne x$, then add a directed path of length $n-1$ from $v_{in}$ to $v_{out}$ with newly introduced $n-2$ internal vertices, and
    \item add the edge $(v_{out}, v_{in})$. 
\end{enumerate} 
The transformation for $v \ne x$ is depicted in \Cref{fig:reduction-long-path-2}.
The special treatment of vertex $x$ guarantees that the size of $V(G')$ is a multiplicity of $n$.
We introduce $n+k$ vertices in place of $x \in X$ and $n$ vertices in place of every other $v \in X$.
These numbers sum up to $(n+k) + n(k-1) = nk + k$.
Adding the number of untouched vertices $v \in V(G) \setminus X$ gives $nk + k + (n-k) = n(k+1)$.

\begin{figure}
\begin{center}
\includegraphics[scale=0.35]{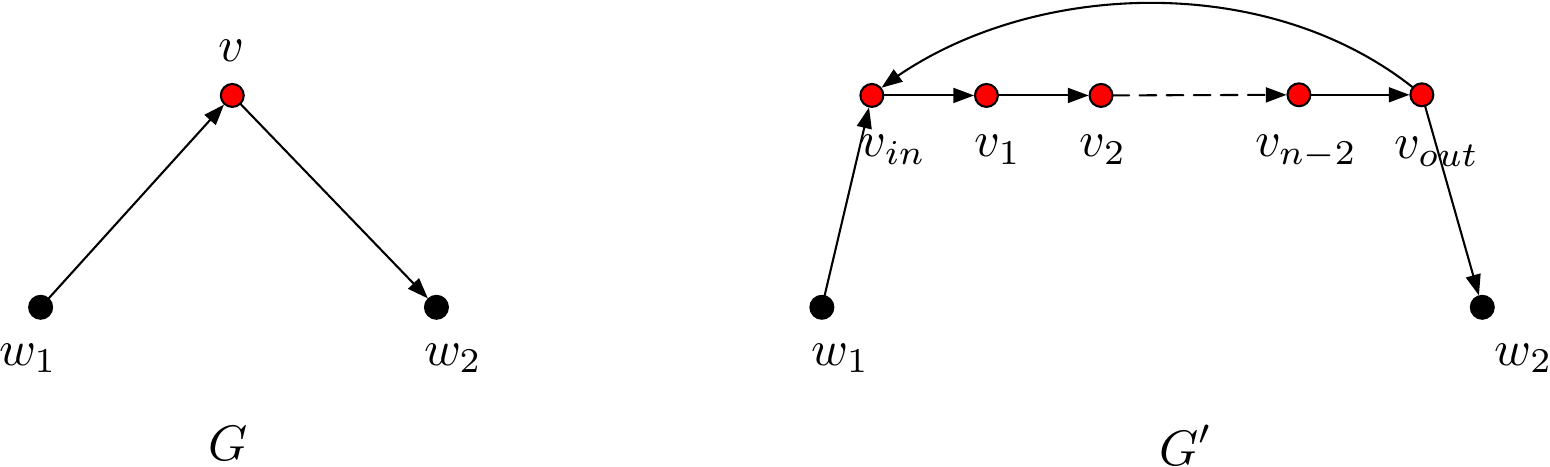}  
\end{center}
\caption{Replacing a vertex  $v \in X$ with a cycle of length $n$ in the construction of $G'$.
The vertices $w_1,w_2$ are exemplary in-neighbor and out-neighbor, respectively, of $v$.
}\label{fig:reduction-long-path-2} 
\end{figure}

\begin{observation}
    \label{lem:girth-size}
    The graph $G'$ has $n(k+1)$ vertices.
\end{observation}

\noindent Since every cycle in $G$ must visit at least one vertex in $X$, its counterpart in $G'$  becomes~long.

\begin{lemma}\label{lem:girth-girth} 
The girth of $G'$ is $n$.
\end{lemma}
\begin{proof}
Consider $v\in X$ different than $x$ (it exists because $|X| \ge 2)$.
The length of the cycle formed by the unique path from $v_{in}$ to $v_{out}$ and edge $(v_{out}, v_{in})$ is $n$. We argue that all the other cycles have length at least $n$. 
Let $X'$ denote the union of all newly introduced vertices in $G'$.
Clearly, any cycle in $G'$ that is contained in $G'[X']$ has length at least $n$.
The graph $G'-X'$ is isomorphic to $G-X$ and hence acyclic.
The only vertices in $X'$ 
that have an in-neighbor in $V(G') \setminus X'$ are $v_{in}$ for $v \in X$.
Thus, any cycle $C$ in $G'$ that visits a vertex in $V(G') \setminus X'$ must visit 
a vertex of the form $v_{in}$ for $v \in X$ and then follow the path from $v_{in}$ to $v_{out}$.
Hence $|V(C)| \ge n$.
This concludes the proof.
\end{proof}

\begin{lemma}\label{lemma:hampath-equiv} 
$G$ has a Hamiltonian cycle if and only if $G'$ has a Hamiltonian cycle.
\end{lemma}
\begin{proof}
    When $C$ is a Hamiltonian cycle in $G$, then it easy to transform it into a Hamiltonian cycle $C'$ in $G'$
    we replacing the occurrence of each vertex $v \in X$ on $C$ with the the unique path from $v^i_{in}$ to $v^i_{out}$.
    In the reverse direction, suppose that $C'$ is a Hamiltonian cycle in $G'$.
    For each  $v \in X$ the vertex $v^i_{in}$ has exactly one outgoing edge, being the first edge on the path $P^i$ from $v^i_{in}$ to $v^i_{out}$.
    Hence $P^i$ is a subpath of $C'$.
    By contracting each such subpath $P_i$ into a single vertex, we transform $C'$ into a cycle in $G$ that visits all the vertices.
    This gives the second implication.
\end{proof}

It follows that $G$ has a Hamiltonian cycle if and only if $G'$ has a cycle of length $|V(G')| = n(k+1)$.
Since $n$ is the girth of $G'$, we get that $(G',k+1)$ is a Yes-instance of \LongCycleG{} if and only if $(G,X)$ is a Yes-instance of \HamDFVS.
This yields \Cref{thm:LongCycleGirth}.

\section{Hardness of the Path Variants}

We justify that the W[1]-hardness results can be transferred to the path variants of both problems.

Consider a digraph $G$ with a vertex $v \in V(G)$.
We obtain the digraph $G^v$ from $G$ by removing $v$ and introducing two new vertices $v_{in}, v_{out}$.
Next, for every $(v,u) \in E(G)$ we add the edge $(v_{out},u)$ to $E(G^v)$, and for every $(u,v) \in E(G)$ we add the edge $(u, v_{in})$ to $E(G^v)$.

\begin{lemma}\label{lem:hamcycle-hampath}
    $G$ has a Hamiltonian cycle if and only if $G^v$ has a Hamiltonian path.
\end{lemma}
\begin{proof}
    The first implication is trivial.
    To see the second one, consider a Hamiltonian path $P'$ in $G^v$.
    Since $v_{out}$ has no incoming edges, it must be the first vertex on $P'$.
    Similarly, $v_{out}$ has no outgoing edges and it is the last vertex on $P'$.
    Let $P$ be a path in $G$ obtained from $P'$ by removing its first and last vertex.
    Then the first vertex on $P$ is an out-neighbor of $v$ in $G$ and the last vertex on $P$ is an in-neighbor of $v$.
    Therefore, $P$ can be turned into a Hamiltonian cycle in $G$ by appending $v$.
\end{proof}

By combining \Cref{thm:HamCycleDFVS} with \Cref{lem:hamcycle-hampath} with immediately obtain the following.

\begin{corollary}
{\sc Hamiltonian Path By DFVS} is $W[1]$-hard.
\end{corollary}

We can also apply this modification to the digraph $G'$ constructed in \Cref{sec:LongPathG} for the purpose of proving W[1]-hardness of \LongCycleG{}.
Recall that we start with a digraph $G$ with $n$ vertices and a DFVS $X$ of size $k$.
Then $G'$ has $n(k+1)$ vertices and girth $n$ (\Cref{lem:girth-girth}).
Also, $G'$ is Hamiltonian if and only if $G$ is Hamiltonian (\Cref{lemma:hampath-equiv}).
We choose any vertex $v \in V(G')$ outside the DFVS of $G'$ and construct digraph $(G')^v$ via \Cref{lem:hamcycle-hampath} that has girth $n$ and  $n(k+1) + 1$ vertices.
Next, we append a directed path on $n-1$ vertices to $v_{in}$
so that now the last vertex of this path has no outgoing edges and the number of vertices increases to $n(k+2)$.
Then $((G')^v,k+2)$ is an instance of \LongPathG{}  equivalent to the instance $(G,X)$ of \HamDFVS.

\begin{corollary}
\LongPathG{} is $W[1]$-hard.
\end{corollary}

\section{XP Algorithm for {\LongPathG}}\label{section:fpt-long-path-above-girth}

Johnson et al.~\cite{DBLP:journals/jct/JohnsonRST01} proved that {\sc Hamiltonian Path} is in {XP} parameterized by directed treewidth.
This result was later extended by de Oliveira Oliveira~\cite{Oliveira16} to capture a wider range of problems expressible in Monadic Second Order (MSO) logic with counting constraints.
As a special case of this theorem, we have the following.

\begin{theorem}[\cite{Oliveira16}]\label{theorem:longest-path-directed-treewidth}
{\sc Longest Path} on directed graphs is in {\em XP} when parameterized by directed treewidth.
\end{theorem}

We will use the Directed Grid Theorem (Theorem \ref{theorem:directed-grid-theroem}), which either returns a cylindrical wall of order $\OO(k)$ in $G$ or concludes that the directed treewidth of $G$ is bounded in terms of~$k$. In the former case, we prove that a path of length $g \cdot k$ always exists in $G$. In the latter case, we use Theorem \ref{theorem:longest-path-directed-treewidth} to solve the problem optimally by an {XP} algorithm.

Throughout this section we abbreviate $W = W_{2k+1}$ (the cylindrical wall of order $2k+1$).
In order to establish correctness of
the algorithm, 
we need to argue that any subdivision $H$ of $W$ admits a path $k$ times longer than the girth of $H$.

We denote the $2k+1$ vertex disjoint cycles in $W$ by $C^W_1, C^W_2, \ldots C^W_{2k+1}$ counting from the innermost one.
For a subdivision $H$ of $W$ we denote the counterpart of $C^W_i$ in $H$ as $C^H_i$.
We refer to the girth of $H$ as $g$.
Note that $g$ lower bounds the length of each cycle $C^H_i$.

\begin{definition}
\label{definition:cycle-segment}
 A subpath in $C^W_i$ is called a {\em segment} of $C^W_i$ if its endpoints have out-neighbors in $C^W_{i+1}$ and none of its internal vertices has out-neighbors in $C^W_{i+1}$. \mic{A subpath in $C^H_i$ is a {\em segment} of $C^H_i$ if it is a subdivision of a segment in $C^W_i$.}
\end{definition}

We make note of a few properties of segments.

\begin{observation}\label{observation:ci-segments}
For every $i \in [2k+1]$ the following hold.
\begin{enumerate}
    \item Each cycle $C^W_i$ is a cyclic concatenation of $k+1$ segments of length $4$.
    \item Each cycle $C^H_i$ is a cyclic concatenation of $k+1$ segments.
    \item If $i>1$ then each segment of $C^W_i$ has a unique internal vertex with an in-neighbor~in~$C^W_{i-1}$.
\end{enumerate}  
\end{observation}

We show that in every cycle $C^H_i$ we can choose a segment $S$ so that the path obtained by 
traversing all vertices in $C^H_i$ that are not internal vertices of $S$ is almost as long as  $C^H_i$.

\begin{lemma}\label{lemma:segment-length} 
For a segment $S$ of $C^H_i$, let $\widehat S$ be the subpath  of $C^H_i$ that starts at the end of $S$ and ends at the start of $S$.
Then there exists a segment $S^H_i$ of $C^H_i$ such that $\widehat{S}^H_i$ has length at least
$g - \frac{g}{k+1}$.
\end{lemma}
\begin{proof}
By Observation \ref{observation:ci-segments} we know that  $C^H_i$
can be partitioned into
$k+1$ segments.
By a counting argument, there exist a segment $S^H_i$ of size at most $\frac{|E(C^H_i)|}{k+1}$. Thus, for the path $\widehat{S}^H_i$ complementing $S^H_i$, we have $|E(\widehat{S}^H_i)| \geq |E(C^H_i)| \cdot (1 - \frac{1}{k+1}) \ge g \cdot (1 - \frac{1}{k+1})$.
\end{proof}

Let $s^\out_i$ denote the first vertex on  $S^H_i$.
If $i > 1$, we define $s^\inn_i$
to be the unique internal vertex of $S^H_i$ 
\mic{corresponding to a vertex in $C^W_i$
with an  in-neighbor in $C^{W}_{i-1}$.}
If $i=1$, we define $s^\inn_i$ to be the last vertex on $S^H_i$.

\begin{figure}
\centering
\includegraphics[scale=0.45]{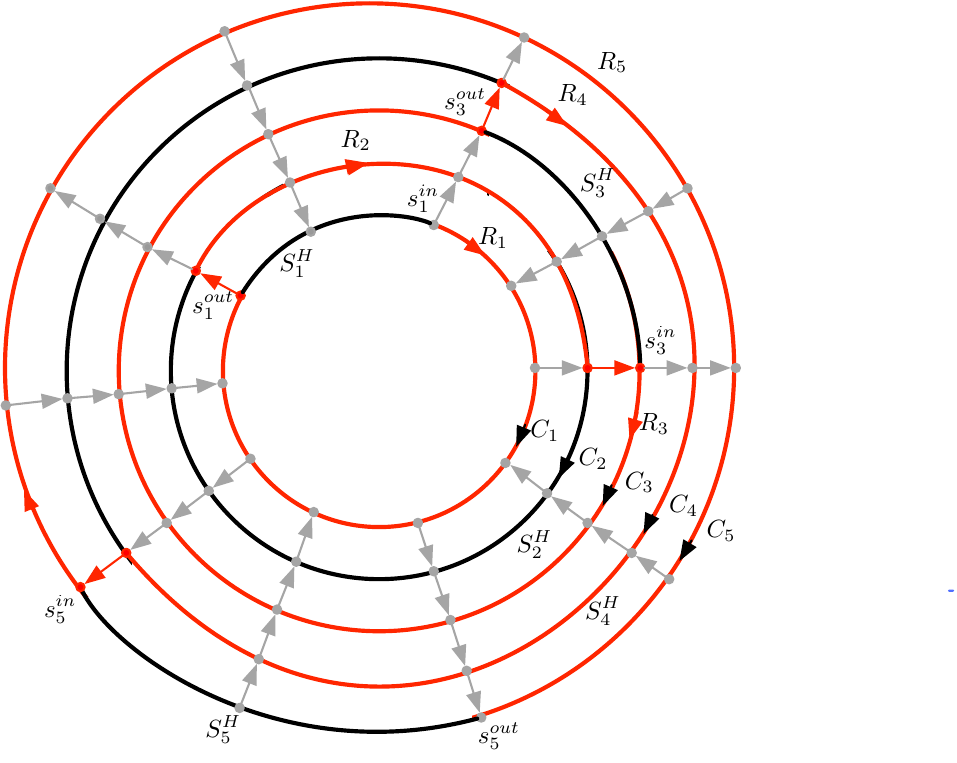}  
\caption{The path $R$ (colored in red) in a subdivision of $W_5$, described in Lemma \ref{lemma:path-length-gk}.}\label{figure-lemma-cylin-grid}  
\end{figure}

\begin{lemma}\label{lemma:path-length-gk}
Let $H$ be a subdivision of $W_{2k+1}$ and $g$ be the girth of $H$.
Then there exists a path of length at least $g \cdot k$ in $H$.
\end{lemma}
\begin{proof}
For $i \in [2k+1]$ we define path $R_i$ in $C_i$ as follows.
If $i$ is odd, we set $R_i$ to be
the subpath of $C_i$ from $s^\inn_{i}$ to $s^\out_{i}$.
If $i$ is even, the path $R_i$ 
begins at the unique vertex in $C^H_i$
\mic{that can be reached via a subdivided edge from 
$s^\out_{i-1}$ and  $R_i$ 
ends at the unique vertex in $C^H_i$
from which $s^\inn_{i+1}$ is reachable by via a subdivided edge.}
Note that the first and last indices in $[2k+1]$ are odd, so the paths $R_i$ are well-defined.
We can now concatenate $R_1,R_2,\dots,R_{2k+1}$ 
\mic{using the subdivided edges between the cycles} 
to obtain a path $R$ in $H$; see Figure \ref{figure-lemma-cylin-grid}.

We argue that $R$ is sufficiently long.
There are $k+1$ odd indices in $[2k+1]$, and, for each of them, $R_i$ 
contains the path $\widehat{S}^H_i$ (from \Cref{lemma:segment-length})
of length at least $g - \frac{g}{k+1}$. 
Since these form vertex disjoint subpaths of $R$, we conclude that $E(R) \ge (g - \frac{g}{k+1})(k+1) = g(k+1) - g = g \cdot k$.
%
\end{proof}

%
%
%

We are ready to summarize the entire algorithm.

\thmLongPathXP*
\begin{proof}
We first execute the algorithm from Theorem \ref{theorem:directed-grid-theroem}, which returns either a directed tree decomposition of width $f(k)$, for some computable function $f$, or a
subgraph $H$ of $G$ that is a subdivision of $W_{2k+1}$.
Note that the girth of $H$ is at least the girth of $G$.
In the first case, we apply the algorithm from Theorem \ref{theorem:longest-path-directed-treewidth} 
to find the longest path in $G$ in polynomial time for fixed $k$. 
In the latter case, Lemma \ref{lemma:path-length-gk} asserts
that $H$ contains a path of length at least $k$ times the girth of $H$, so
we can report that $(G,k)$ a Yes-instance.
\end{proof}


\bibliography{main-file.bib}

\end{document}